\documentclass{llncs}
\usepackage{graphicx}
\usepackage{url}
\usepackage{latexsym}
\usepackage{amscd}
\usepackage{amsfonts}
\usepackage{amssymb}
\usepackage{mathtools}
\usepackage[usenames,dvipsnames,svgnames,table]{xcolor}
\usepackage{hyperref}

\usepackage{varwidth}

\usepackage{float}
\usepackage{longtable}
\usepackage{tikz}
\usetikzlibrary{arrows, chains, matrix, positioning, scopes, patterns, shapes}
\usepackage{mathrsfs}
\usepackage{array}
\usepackage{bbm}
\usepackage{algorithm}
\usepackage{algpseudocode} 

\usepackage{multirow}

\newcounter{algsubstate}

\newenvironment{algsubstates}
{\setcounter{algsubstate}{0}%
	\renewcommand{\State}{%
		\refstepcounter{algsubstate}%
		\Statex {\alph{algsubstate}:}\space}}
{}

\usepackage[braket, qm]{qcircuit}

\usetikzlibrary{decorations.pathreplacing}

\usepackage{caption} 
\usepackage{enumitem}

\mathchardef\hyphen="2D

\usepackage{multirow}
\usepackage{multicol} 

\newcommand*\SmallTextScript[2]{{\mathchoice{\displaystyle #2}
		{\textstyle #2}
		{\scalebox{#1}{\ensuremath{\scriptstyle #2}}}%
		{\scalebox{#1}{\ensuremath{\scriptscriptstyle #2}}}%
}}

\newcommand*{\poly}{\ensuremath{\mathrm{poly}}}
\newcommand*{\eps}{\ensuremath{\varepsilon}}

\newcommand{\R}{{{\mathbb R}}}

\newcommand{\GF}{\ensuremath{\mathbb{F}_2}}

\newcommand{\bigO}{\mathcal{O}}

\newcommand{\smallo}{o} 

\newcommand*\PROB\Pr 
\DeclareMathOperator*{\EXPECT}{\mathbb{E}}

\newcommand{\HMat}{\ensuremath{\mathbf{H}}}
\newcommand{\QMat}{\ensuremath{\mathbf{Q}}}
\newcommand{\Id}{\ensuremath{\mathbf{I}}}
\newcommand{\ZeroM}{\textbf{0}} 

\newcommand{\evec}{\ensuremath{\mathbf{e}}}
\newcommand{\svec}{\ensuremath{\mathbf{s}}}
\newcommand{\vvec}{\ensuremath{\mathbf{v}}}
\newcommand{\zvec}{\ensuremath{\mathbf{0}}}
\newcommand{\xvec}{\ensuremath{\mathbf{x}}}
\newcommand{\cvec}{\ensuremath{\mathbf{c}}}
\newcommand{\qvec}{\ensuremath{\mathbf{q}}}
\newcommand{\pvec}{\ensuremath{\mathbf{p}}}
\newcommand{\tvec}{\ensuremath{\mathbf{t}}}

\newcommand{\RepMMT}{\ensuremath{\mathcal{R}_{\protect\SmallTextScript{0.70}{\texttt{MMT}}}}}
\newcommand{\RepBJMM}{\ensuremath{\mathcal{R}_{\protect\SmallTextScript{0.70}{\texttt{BJMM}}}}}

\renewcommand{\L}{\ensuremath{\mathcal{L}}}

\newcommand{\C}{\ensuremath{\mathcal{C}}} 
\newcommand{\D}{\ensuremath{\mathcal{D}}} 
\newcommand{\Region}{\ensuremath{\text{Region}}} 
\newcommand{\Bucket}{\ensuremath{\text{Bucket}}} 

\renewcommand{\P}{\ensuremath{\mathcal{P}}}

\newcommand*{\Lout}{\ensuremath{\L_{\mkern-0.5mu\protect\SmallTextScript{0.85}{\textup{out}}}}}
\newcommand*{\Sout}{\ensuremath{S_{\mkern-0.5mu\protect\SmallTextScript{0.85}{\textup{out}}}}}
\newcommand{\wt}{\ensuremath{\mathit{wt}}} 
\newcommand{\dist}{\ensuremath{\mathit{dt}}} 

\newcommand*{\softO}{\widetilde{\bigO}}

\newcommand{\const}{\mathsf{c}}

\newcommand{\transpose}{\mkern0.7mu^{\mathsf{ t}}}

\newcolumntype{L}{>{\varwidth[l]{0.55\linewidth}}l<{\endvarwidth}}

\pgfdeclarepatternformonly{my dots}{\pgfqpoint{-1pt}{-1pt}}{\pgfqpoint{2.0pt}{2.0pt}}{\pgfqpoint{2pt}{2pt}}%
{
	\pgfpathcircle{\pgfqpoint{0pt}{0pt}}{.35pt}
	\pgfpathcircle{\pgfqpoint{1pt}{1pt}}{.35pt}
	\pgfusepath{fill}
}

\tikzset{
	master/.style={
		execute at end picture={
			\coordinate (lower right) at (current bounding box.south east);
			\coordinate (upper left) at (current bounding box.north west);
		}
	},
	slave/.style={
		execute at end picture={
			\pgfresetboundingbox
			\path  (lower right)rectangle (upper left) ;
		}
	}
}

\begin{document}

\def\May{Alexander May}
\def\MayEmail{alex.may@{\allowbreak}rub.{\allowbreak}de}
\def\MayDept{Faculty of Mathematics}
\def\MayUniversity{Ruhr University Bochum}
\def\MayInstitute{Horst G{\"o}rtz Institute for IT-Security}
\def\MayCountry{Germany}

\def\Kirshanova{Elena Kirshanova}
\def\KirshanovaEmail{elena.kirshanova@{\allowbreak}ens-lyon.{\allowbreak}fr}
\def\KirshanovaUniversity{Laboratoire LIP, ENS de Lyon}
\def\KirshanovaCountry{France}

\author{
      \Kirshanova{} 
}

\institute{
	\KirshanovaUniversity \\
	\KirshanovaEmail \\
}

\def\titletext{Improved Quantum Information Set Decoding}
\def\shorttitletext{}

\def\abstracttext{
	In this paper we present quantum information set decoding (ISD) algorithms for binary linear codes. First, we give an alternative view on the quantum walk based algorithms proposed by Kachigar and Tillich (PQCrypto'17). It is more general and allows to consider any ISD algorithm that has certain properties. The algorithms of May-Meuer-Thomae and Becker-Jeux-May-Meuer satisfy these properties.  Second, we translate May-Ozerov Near Neighbour technique (Eurocrypt'15) to an `update-and-query' language more suitable for the quantum walk framework. First, this re-interpretation makes possible to analyse a broader class of algorithms and, second, allows us to combine Near Neighbour search with the quantum walk framework and use both techniques to give a quantum version of Dumer's ISD with Near Neighbour.
	
}

\def\keywordstext{\textbf{Keywords} 
	information set decoding, quantum walk, near neighbour.
}

\date{\vspace{-5ex}} 
\title{\titletext}
\maketitle
 \begin{abstract}
 	
  \abstracttext
  \noindent \keywordstext
 \end{abstract}
  

\section{Introduction}\label{sec:Intro}

%
%
%
%
%

\emph{The Information Set Decoding problem} with integer parameters $n, k, d$
asks to find the error-vector $\evec \in \GF^{n}$ given a matrix $\HMat \in \GF^{(n-k) \times n}$ and a vector $\svec = \HMat \evec\transpose$ such that  the Hamming weight of $\evec$, denoted $w:=\wt(\evec)$, is bounded by some integer. The matrix $\HMat$ is called the parity-check matrix of a binary linear $[n, k, d]$-code $\mathcal{C}$, where $d$ is the minimum distance of the code.
In this work, we stick to the so-called \emph{full distance decoding} setting, i.e., when we search for $\evec$ with $\wt(\evec) \leq d$. The analysis is easy to adapt to \emph{half-distance decoding}, i.e., when $\wt(\evec) \leq \lfloor \frac{d-1}{2} \rfloor$.  

The ISD problem is relevant not only in coding theory but also in cryptography: several cryptographic constructions, e.g. \cite{McEliece}, rely on the hardness of ISD. The problem seems to be intractable even for quantum computers, which makes these constructions attractive for post-quantum cryptography. 

First classical ISD algorithm due to Prange dates back to 1962 \cite{Prange} followed by a series of improvements \cite{Stern,Dumer,FinSen09,MMT11,BJMM12}, culminating in algorithms \cite{MO15,BM17,BM18} that rely on Nearest Neighbour techniques in Hamming metric. On the quantum side, the ISD problem received much less attention: Bernstein in \cite{Bern10} analysed a quantum version of Prange's algorithm, and recently Kachigar and Tillich \cite{KT17} gave a series of ISD algorithms based on quantum walks. Our results extend the work of \cite{KT17}.

\paragraph{Our contributions:}
\begin{enumerate} 
	\item We present another way of analysing quantum ISD algorithms from \cite{KT17}: it allows to simplify the complexity estimates for every ISD algorithm given in \cite{KT17};
	\item We re-phrase May-Ozerov Near Neighbour algorithm  \cite{MO15} in the `update-and-query' language and give a method to analyse its complexity;
	\item We present a quantum version of the May-Ozerov ISD algorithm.
\end{enumerate}

Our second contribution is of independent interest as it provides an alternative but more flexible view on May-Ozerov Near Neighbour algorithm for the Hamming metric. We give simple formulas for analysing its complexity which allow us to stay in the Hamming space, i.e., without reductions from other metrics as it is usually done in the literature \cite{Chr17}.
The third contribution answers the problem left open in \cite{KT17}, namely, how to use Near Neighbour technique within quantum walks. Our results are summarized in the table below.

\vspace{-10pt}
\renewcommand{\arraystretch}{1.1}
\setlength{\tabcolsep}{7pt}
\begin{table}[H]
	\centering
		\begin{tabular}{|l|c|c|c|c|} \hline
			\multirow{2}{*}{\textbf{Algorithm}} & \multicolumn{2}{c|}{ \textbf{Quantum} }  & \multicolumn{2}{c|}{\textbf{Classical}} \\ \cline{2-5}
			& \textbf{Time} & \textbf{Space}  & \textbf{Time} & \textbf{Space}   \\ \hline
			Prange \cite{Bern10,Prange} & 0.060350& $--$ & 0.120600 & $--$  \\ \hline
			Stern/Dumer \cite{Stern,Dumer}  &  &  & \multirow{3}{*}{0.116035} & \multirow{3}{*}{0.03644} \\ 
			\hspace*{5pt} + Shamir-Schroeppel (SS) \cite{KT17} & 0.059697 & 0.00618 & & \\ 
			\hspace*{5pt} + Near Neighbour (NN) Sect.\ref{sec:DecodingWithNN} & 0.059922 & 0.00897 & 0.113762 & 0.04248 \\ 
			\hspace*{5pt} + SS + NN Sect.\ref{sec:DecodingWithNN} & 0.059450 & 0.00808 & & \\ \hline
			MMT \cite{MMT11} &  & &  \multirow{3}{*}{0.111468} &  \multirow{3}{*}{0.05408} \\
			\hspace*{5pt} -- Kachigar-Tillich \cite{KT17} & 0.059037  & 0.01502 & &    \\
			BJMM \cite{BJMM12}& &  & \multirow{3}{*}{0.101998} &  \multirow{3}{*}{0.07590} \\ 
			\hspace*{5pt} -- Kachigar-Tillich \cite{KT17} & 0.058696 & 0.01877 & & \\
		 \hline
		\end{tabular}
	\caption[Runtimes]{Running time and space complexities of ISD algorithms (full distance decoding). The columns give the exponent-constants $\const$, i.e., runtime and memory complexities are of the form $\bigO(2^{\const n})$. For Prange's algorithm, the space is $\poly(n)$. 
	}
	\label{table:RunTimes}
\end{table}
\vspace{-20pt}

For each classical algorithm, Table~\ref{table:RunTimes} gives running time and space complexities of their quantum counterparts. By the `quantum space' in Table~\ref{table:RunTimes}, we mean the number of qubits in a quantum state an algorithm operates on. Note that this work does improve over Kachigar-Tillich quantum versions of MMT or BJMM ISD algorithms, but we present a different way of analysing the asymptotic complexities these algorithms.

In Sect.~\ref{sec:DecodingWithNN} we show how to combine the Near Neighbour search of May and Ozerov \cite{MO15} with quantum version of the ISD algorithm due to Dumer \cite{Dumer}. Combined with the so-called Shamir-Schroeppel trick \cite{ShamirSchro81}, which was already used in \cite{KT17}, we can slightly improve the running time of this algorithm. 
 
We note that, as in the classical setting, the Near Neighbour technique requires more memory, but we are still far from the Time=Memory regime. It turns out that, as opposed to the classical case, quantum Near Neighbour search does not improve MMT or BJMM. We argue why this is the case at the and of Sect.~\ref{sec:DecodingWithNN}. We leave as an open problem an application of quantum Near Neighbour to MMT/BJMM algorithms as well as quantum speed-ups for algorithm described in the recent work be Both-May \cite{BM17,BM18}.


\section{Preliminaries}\label{sec:Prelims}

We start with overview on classical algorithms for ISD, namely, Prange \cite{Prange}, Stern and its variants \cite{Stern,Dumer}, MMT \cite{MMT11}, and BJMM \cite{BJMM12} algorithms. We continue with known quantum speed-ups for these algorithms. 

\subsection{Classical ISD algorithms}
All known ISD algorithms try to find the error-vector $\evec$ by a clever enumeration of the search space for $\evec$, which is of size $\binom{n}{w} \approx 2^{n \cdot H \left( \tfrac{w}{n}\right)}$, where $H(x) = -x \log x -(1-x)\log(1-x) $ is the binary entropy function. In the analysis of ISD algorithms, it is common to relate the parameters $w$ (the error-weight), and $k$ (the rank of a code) to dimension $n$, and simplify the running times to the form $2^{\const n}$ for some constant $\const$.\footnote{We omit sub-exponential in $n$ factors throughout, because we are only interested in the constant $\const$. Furthermore, our analysis is for an average case and we sometimes omit the word `expected'.} To do this, we make use of Gilbert-Varshamov bound which states that $\tfrac{k}{n} = 1- H\left( \tfrac{w}{n}\right)$ as $n \rightarrow \infty$. This gives us a way to express $w$ as a function of $n$ and $k$. Finally, the running time of an ISD algorithm is obtained by a brute-force search over all $\tfrac{k}{n} \in [0, \tfrac{1}{2}]$ (up to some precision) that leads to the worst-case complexity. In the classical setting, this worst-case is reached by codes of rate $\tfrac{k}{n} \approx 0.447$, while in the quantum regime it is  $\tfrac{k}{n} \approx 0.45$.

Decoding algorithms start by permuting the columns of $\HMat$ which is equivalent to permuting the positions of $1$'s in $\evec$. The goal is to find a permutation $\pi \in S_n$ such that $\pi(\evec)$ has exactly $p \leq w$ $1$'s on the first $k$ coordinates and the remaining weight of $w-p$ is distributed over the last $n-k$ coordinates. All known ISD algorithms make use of the fact that such a permutation is found. We expect to find a good $\pi$ after $\P(p)$ trials, where
\begin{equation} \label{eq:PermutationProbabiltiy0}
\P(p)=\frac{{\binom{k}{p} \binom{n-k}{w-p}}}{ \binom{n}{w}}.
\end{equation}
The choice of $p$ and how we proceed with $\pi(\HMat)$ depends on the ISD algorithm. 

For example, \emph{Prange's} algorithm \cite{Prange} searches for a permutation $\pi $ that leads to $p=0$. To check whether a candidate $\pi$ is good, it transforms $\pi(\HMat)$ into systematic form $[\QMat \mid  \Id_{n-k}]$ (provided the last $n-k$ columns of $\pi(\HMat)$ form an invertible matrix which happens with constant success probability). The same transformation is applied to the syndrome $\svec$ giving a new syndrome $\bar{\svec}$. From the choice of $p$, it is easy to see that for a good $\pi$, we  just `read-off' the error-vector from the new syndrome, i.e.,  $\pi(\evec)  = \bar{\svec}$, and to verify a candidate $\pi$, we check if $\wt(\bar{\svec}) = w$. We expect to find a good $\pi$ after $\P(0)$ trials. 

From now on, we assume that we work with systematic form of $\HMat$, i.e.\
\begin{equation} \label{eq:ISDeq}
[\QMat \mid  \Id_{n-k}] \cdot  \evec= \bar{\svec} \quad \text{for } \; \QMat \in \GF^{n-k \times k}.
\end{equation}
Other than restricting the weight of $\evec$ to be 0 on the last $n-k$ coordinates, we may as well allow $p>0$ at the price of a more expensive check for $\pi$. This is the choice of \emph{Stern's} algorithm \cite{Stern}, which was later improved in \cite{Dumer} (see also \cite{FinSen09}). We describe the improved version.
We start by adjusting the systematic form of $\HMat$ introducing the $\ell$-length 0-window, so that Eq.~\ref{eq:ISDeq} becomes
\begin{equation} \label{eq:ISDeqFZ}
\left[\QMat \Bigg| \frac{\ZeroM}{\Id_{n-k-\ell}}\right] \cdot  \evec= \bar{\svec} \quad \text{for } \; \QMat \in \GF^{n-k \times k+\ell}.
\end{equation}
Now we search for a permutation $\pi$ that splits the error as 
\begin{equation*}
\setlength{\abovedisplayskip}{4pt}
\setlength{\belowdisplayskip}{4pt}
\evec =[\evec_1 || \zvec^{\frac{k+\ell}{2}} || \zvec^{n-k-\ell}]+[\zvec^{\frac{k+\ell}{2}}|| \evec_2 || \zvec^{n-k}]+[\zvec^{k+\ell }|| \evec_3], 
\end{equation*}
such that $\wt(\evec_1)= \wt(\evec_2) = p/2$ and $\wt(\evec_3)=w-p$, where $\evec_i$'s are of appropriate dimensions. With such an $\evec$, we can re-write Eq.~\eqref{eq:ISDeqFZ} as 
\begin{equation} \label{eq:ISDeqExt}
	\QMat \cdot [\evec_1|| \zvec^{\frac{k+\ell}{2} }] + \QMat \cdot   [\zvec^{\frac{k+\ell}{2}} || \evec_2] = \bar{\svec} + \evec_3.
\end{equation}
We enumerate all possible $\smash{\binom{(k+\ell)/2}{p/2}}$ vectors of the form $\vvec_1 := \QMat [\evec_1|| \zvec^{\frac{k+\ell}{2}} ]$ into a list $\L_1$ and all vectors of the form $ \vvec_2 := \QMat [\zvec^{\frac{k+\ell}{2}} || \evec_2] + \bar{\svec}$ into a list $\L_2$. The above equation tells us that for the correct pair $(\evec_1, \evec_2)$, the sum of the corresponding list-vectors equals to $\zvec$ on the first  $\ell$-coordinates. We search for two vectors $\vvec_1 \in \L_1, \vvec_2 \in \L_2$ that are equal on this $\ell$-window. We call such pair $(\vvec_1, \vvec_2)$ a \emph{match}. We check among these matches if the Hamming distance between $\vvec_1, \vvec_2$, denoted $\dist(\vvec_1, \vvec_2)$, is $w-p$. To retrieve the error-vector, we store $\evec_i$'s together with the corresponding $\vvec_i$'s in the lists. 
The probability of finding a permutation that meets all the requirements is
\begin{equation} \label{eq:PermutationProbabiltiyFS}
\setlength{\abovedisplayskip}{4pt}
\setlength{\belowdisplayskip}{4pt}
\P(p, \ell)=\frac{{\binom{k+\ell}{p} \binom{n-k-\ell }{w-p}}}{ \binom{n}{w}}.
\end{equation}
It would be more precise to have $\smash{\binom{(k+\ell)/2}{p/2}^2}$ instead of $\binom{k+\ell}{p}$ in the above formula, but these two quantities differ by only a factor of $\poly(n)$ which we ignore. The expected running time of the algorithm is then
\begin{equation}\label{eq:RTFS}
T = \P(p, \ell)^{-1} \cdot \max \left\{|\L_2|, \frac{|\L_1| \cdot |\L_2|}{ 2^{\ell}}\right\},
\end{equation}
where the first argument of $\max$ is the time to sort $\L_2$, the second is the expected number of pairs from $\L_1 \times \L_2$ that are equal on $\ell$, which we check for a solution.
See Fig.~\ref{fig:SternAndMMT} for an illustration of the algorithm.

\tikzset{align at top/.style={baseline=(current bounding box.north)}}
\begin{figure}[H]
	\centering
	\begin{tikzpicture}[align at top]
	
	\filldraw[fill=none, black] (0.0,0.0) rectangle (3.0,1.5) node[](rect) {};
	\draw[black] (1.6, 1.5) -- (1.6, 0.0) {};
	\node[right] at (0.5, 0.7) {\Large $\mathbf{Q}$};
	\node[right] at (1.7, 0.5) {\large $\mathbf{I}_{n-k-l}$};
	
	\draw[black, thin] (1.6, 1.0) -- (3.0, 1.0);
	\node[right] at (2.1, 1.2) {\large $\mathbf{0}$}; 
	
	\draw [decorate,decoration={brace,amplitude=2.5pt,mirror,raise=1pt}]
	(3.0,1.0) -- (3.0,1.5) node [black,midway,xshift=8pt] {\scriptsize
		$\ell$};
	
	\draw[black, dashed] (0.8, 0.0) -- (0.8, 1.5) node[above] {};
	\draw [decorate,decoration={brace,amplitude=2.0pt,raise=1pt}]
	(0.0,1.5) -- (1.6,1.5) node [black, above left, xshift=-10pt, yshift=0pt] {\scriptsize $k+\ell$};
	
	\draw[black] (0.0, -0.4) rectangle (0.78, -0.2) node[left, xshift=-20pt, yshift=-3pt, color=black]{\scriptsize $\evec_1$};
	
	\draw[black] (0.82, -0.4) rectangle (1.59, -0.2) node[right, xshift=0pt, yshift=-3pt, color=black]{\scriptsize $\evec_2$};
	
	
	\filldraw[pattern=my dots, draw=none] (0.2, 0.0) rectangle (0.4, 1.5) node[above, xshift=-4pt]{\scriptsize $\tfrac{p}{2}$};
	
	\filldraw[pattern=my dots, draw=none] (0.2, -0.4) rectangle (0.4, -0.2) {};
	
	\filldraw[pattern=my dots, draw=none] (1.3, 0.0) rectangle (1.5, 1.5) node[above, xshift=-4pt]{\scriptsize $\tfrac{p}{2}$};
	
	\filldraw[pattern=my dots, draw=none] (1.3, -0.4) rectangle (1.5, -0.2) {};
	
	\filldraw[fill=none, black] (-0.6,-3.0) rectangle (0.6,-0.8) node[midway](L1) {\large $\L_1$};
	\filldraw[fill=none, black] (1.0,-3.0) rectangle (2.2,-0.8) node[midway](L2) {\large $\L_2$};
	
	\draw[-latex'] (0.4, -0.4) to [bend left=-20] (0.0, -0.8);
	\draw[-latex'] (1.2, -0.4) to [bend right=-20] (1.5, -0.8);
	
	\filldraw[pattern=north west lines, pattern color=gray] (-0.6,-1.4) rectangle (0.6,-0.8) node[midway]{};
	\filldraw[pattern=north west lines, pattern color=gray] (1.0,-1.4) rectangle (2.2,-0.8);
	\draw [decorate,decoration={brace,amplitude=2.5pt,mirror,raise=1pt}]
	(2.2,-1.4) -- (2.2,-0.8) node [black, right, xshift=3pt, yshift=-10pt] {$S_2$};
	
	\filldraw[fill=none, black]  (0.1, -5.5) rectangle (1.3, -3.5) node[midway] {$\Lout$};
	
	\filldraw[fill=gray!30, draw=black] (1.0, -5.5) rectangle (1.3, -3.5) {};
	\draw [decorate,decoration={brace,amplitude=2.5pt,mirror,raise=1pt}]
	(1.3,-4.2) -- (1.3,-3.5) node [black, right, xshift=3pt, yshift=-6pt] {\scriptsize $\ket{Aux}$};
	
	\draw [decorate,decoration={brace,amplitude=1.5pt,mirror,raise=1pt}]
	(1.0,-5.5) -- (1.3,-5.5) node [black, below, xshift=-3pt, yshift=-3pt] {\scriptsize $\ell$};
	
	\filldraw[pattern=north west lines, pattern color=gray] (0.1,-4.2) rectangle (1.3,-3.5) {};
	
	\draw[-latex'] (0.0, -3.0) to [bend left=10] (0.7, -3.5);
	\draw[-latex'] (1.6, -3.0) to [bend left=-10] (0.7, -3.5);

	\end{tikzpicture}
	\hspace{10pt}
	\begin{tikzpicture}[align at top]
	\filldraw[fill=none, black] (0.0,0.0) rectangle (3.0,1.5) node[](rect) {};
	\draw[black] (1.6, 1.5) -- (1.6, 0.0) {};
	\node[right] at (0.5, 0.7) {\Large $\mathbf{Q}$};
	\node[right] at (1.7, 0.5) {\large $\mathbf{I}_{n-k-l}$};
	
	\draw[black, thin] (1.6, 1.0) -- (3.0, 1.0);
	\node[right] at (2.1, 1.2) {\large $\mathbf{0}$};
	
	\draw[black, dashed] (0.8, 0.0) -- (0.8, 1.5) node[above] {};
	\draw [decorate,decoration={brace,amplitude=2.0pt,raise=1pt}]
	(0.0,1.5) -- (1.6,1.5) node [black, above left, xshift=-10pt, yshift=0pt] {\scriptsize $k+\ell$};
	
	\draw[black] (0.0, -0.4) rectangle (0.78, -0.2) node[left, xshift=-20pt, yshift=-3pt, color=black]{\scriptsize $\evec_1$};
	
	\draw[black] (0.82, -0.4) rectangle (1.59, -0.2) node[right, xshift=0pt, yshift=-3pt, color=black]{\scriptsize $\evec_2$};
	
	\filldraw[pattern=my dots, draw=none] (0.2, 0.0) rectangle (0.3, 1.5) node[above, xshift=-2pt]{\scriptsize $\tfrac{p}{4}$};
	
	\filldraw[pattern=my dots, draw=none] (0.2, -0.4) rectangle (0.3, -0.2) {};
	
	\filldraw[pattern=my dots, draw=none] (1.3, 0.0) rectangle (1.4, 1.5) node[above, xshift=-2pt]{\scriptsize $\tfrac{p}{4}$};
	
	\filldraw[pattern=my dots, draw=none] (1.3, -0.4) rectangle (1.4, -0.2) {};
	
	\filldraw[fill=none, black] (-2.0,-2.5) rectangle (-0.8,-0.8) node[midway](L11) {\large $\L_{1,1}$};
	\filldraw[fill=none, black] (-0.6,-2.5) rectangle (0.6,-0.8) node[midway](L12) {\large $\L_{1,2}$};
	
	\filldraw[fill=none, black] (0.9,-2.5) rectangle (2.1,-0.8) node[midway] {\large $\L_{2,1}$};
	\filldraw[fill=none, black] (2.3,-2.5) rectangle (3.5,-0.8) node[midway] {\large $\L_{2,2}$};

	\filldraw[pattern=north west lines, pattern color=gray] (-2.0,-1.2) rectangle (-0.8,-0.8) {};
	\filldraw[pattern=north west lines, pattern color=gray] (-0.6,-1.2) rectangle (0.6,-0.8) {};
	\filldraw[pattern=north west lines, pattern color=gray] (0.9,-1.2) rectangle (2.1,-0.8) {};
	\filldraw[pattern=north west lines, pattern color=gray] (2.3,-1.2) rectangle (3.5,-0.8) {};
	
	\draw [decorate,decoration={brace,amplitude=2.5pt,mirror,raise=1pt}]
	(3.5,-1.2) -- (3.5,-0.8) node [black, right, xshift=3pt, yshift=-7pt] {$S_{2,2}$};
	
	\draw[-latex'] (0.4, -0.42) to [bend left=-10] (-1.4, -0.8);
	\draw[-latex'] (1.2, -0.42) to [bend right=10] (0.0, -0.8);
	
	\draw[-latex'] (0.4, -0.42) to [bend left=10] (1.6, -0.8);
	\draw[-latex'] (1.2, -0.42) to [bend right=-10] (2.9, -0.8);
	
	\filldraw[fill=none, black, line width=0.2mm] (-1.4,-4.5) rectangle (-0.2,-3.0) node[midway](L1) {\large $\L_{1}$};
	\filldraw[fill=gray!30, draw=black] (-0.5, -4.5) rectangle (-0.2, -3.0) {};
	\filldraw[pattern=north west lines, pattern color=gray] (-1.4,-3.4) rectangle (-0.2, -3.0) {};
	
	\filldraw[fill=none, black, black, line width=0.2mm] (1.5,-4.5) rectangle (2.7,-3.0) node[midway](L2) {\large $\L_{2}$};
	\filldraw[fill=gray!30, draw=black] (2.4, -4.5) rectangle (2.7,-3.0) {};
	\filldraw[pattern=north west lines, pattern color=gray] (1.5,-3.4) rectangle (2.7,-3.0) {};
	
	\draw [decorate,decoration={brace,amplitude=2.5pt,mirror,raise=1pt}]
	(2.7,-3.4) -- (2.7,-3.0) node [black, right, xshift=3pt, yshift=-6pt] {\scriptsize $\ket{Aux}$};
	
	\draw [decorate,decoration={brace,amplitude=1.5pt,mirror,raise=1pt}]
	(2.4, -4.5) -- (2.7, -4.5) node [black, below, xshift=-3pt, yshift=-3pt] {\scriptsize $\RepMMT$};
	
	\draw[-latex'] (-1.4, -2.5) to [bend left=10] (-0.8, -3.0);
	\draw[-latex'] (0.0, -2.5) to [bend right=10] (-0.8, -3.0);
	
	\draw[-latex'] (1.6, -2.5) to [bend left=10] (2.2, -3.0);
	\draw[-latex'] (2.9, -2.5) to [bend right=10] (2.2, -3.0);
	
	\filldraw[fill=none, black] (0.1,-6.5) rectangle (1.3,-5.0) {};
	
	\draw[-latex'] (-0.8, -4.5) to [bend left=10] (0.6, -5.0);
	\draw[-latex'] (2.2, -4.5) to [bend right=10] (0.6, -5.0);
	
	\filldraw[fill=gray!20, draw=black] (1.0, -6.5) rectangle (1.3,-5.0) {};
	\filldraw[fill=gray!20, draw=black] (0.7, -6.5) rectangle (1.0,-5.0) {};
	
	\node[] at (0.7, -5.7) {\large $\Lout$};
	\filldraw[pattern=north west lines, pattern color=gray] (0.1,-5.4) rectangle (1.3,-5.0) {};
	
	\draw [decorate,decoration={brace,amplitude=2.5pt,mirror,raise=1pt}]
	(1.3,-5.4) -- (1.3,-5.0) node [black, right, xshift=3pt, yshift=-6pt] {\scriptsize $\ket{Aux}$};
	
	\draw [decorate,decoration={brace,amplitude=1.5pt,mirror,raise=1pt}]
	(0.7, -6.5) -- (1.3, -6.5) node [black, below, xshift=-8pt, yshift=-3pt] {\scriptsize $\ell$};
	
	\end{tikzpicture}
	\caption{\emph{On the left}: A variant of Stern's ISD algorithm due to Dumer  \cite{Dumer}. The list $\L_1$ is constructed from all possible $p/2$-weight vectors $\evec_1 \in \GF^{(k+\ell)/2}$:  $\L_1 = \{(\evec_1,\QMat \cdot  [\evec_1|| \zvec^{(k+\ell)/2}]) \}$. $\L_2$ is constructed similarly with $\zvec^{(k+\ell)/2}$ and $\evec_1$ swapped. Gray-shaded vertical strip indicates the coordinates on which the elements $\vvec_1 \in \L_1$ and $\vvec_2 \in \L_2$ match. Line-shaded horizontal strips indicate a subset of lists stored on quantum registers during the execution of quantum walk search algorithm. 
		\newline
		\emph{On the right:} May-Meurer-Thomae decoding \cite{MMT11}. The lists $\L_1, \L_2$ are shorter than in Dumer's algorithm as their elements already match on $\RepMMT$-coordinates. 
		Quantum walk operates on subsets $S_{i.j}$ of the bottom lists $L_{i,j}$. We also keep the auxiliary register $\ket{Aux}$, where we store the result of merging $S_{i,j}$ into $S_{i} \subset \L_i$, and $\Sout \subset \Lout.$ 
	}
	\label{fig:SternAndMMT}
\end{figure}

\emph{The Representation technique} of \cite{BJMM12,MMT11} further improves the search for matching vectors by constructing the lists $\L_1, \L_2$ faster. Now the list $\L_1$ consists of all pairs $(\evec_1,\QMat \evec_1)$ where $\evec_1 \in \GF^{k+\ell}$ (as opposed to $\evec_1 \in \GF^{(k+\ell)/2}$) with $\wt(\evec_1) = p/2$. Similarly,  $\L_2=\{(\evec_2,\QMat \evec_2) \; | \; \forall \evec_2\in \GF^{k+\ell}, \wt(\evec_2) = p/2 \}$. The key observation is that now there are $\RepMMT:=\binom{p}{p/2}$ ways to represent the target $\evec$ as $\evec = \evec_1 + \evec_2$. Hence, it is enough to construct only an $ \RepMMT-$fraction of $\L_1, \L_2$. Such a fraction of $\L_1$ (analogously, for $\L_2$) is built by merging in the meet-in-the-middle way yet another two lists $\L_{1,1}$ and $\L_{1,2}$ filled with vectors of the form $\QMat [\evec_{1,1} || \zvec^{(k+\ell)/2}]$ (for $\L_{1,1}$)  and $\QMat [\zvec^{(k+\ell)/2} || \evec_{1,2} ]$ (for $\L_{1,2}$) for all  $p/4$-weight $\evec_{1,1}$ and $\evec_{1,2}$, respectively. These starting lists are of size
\begin{equation} \label{eq:MMTListSize}
\setlength{\abovedisplayskip}{4pt}
\setlength{\belowdisplayskip}{4pt}
	|L_{i,j}| = \binom{(k+\ell)/2}{p/4}, \; i,j \in \{1,2\}.
\end{equation}

During the merge, we force vectors from $\L_{1,1}$ be equal to vectors from $\L_{1,2}$ on $\log \RepMMT$ coordinates leaving only one (in expectation) pair ($\evec_1, \evec_2) \in \L_1 \times \L_2$  whose sum gives $\evec$ (see Fig.~\ref{fig:SternAndMMT}, right). Here and later, we shall abuse notations slightly: technically, the list elements are pairs $(\evec, \QMat \evec)$, but the merge is always done on the second element, and the error retrieval is done on the first.

The number of necessary permutations we need to try is given by Eq.~(\ref{eq:PermutationProbabiltiyFS}). Provided a good $\pi$ is found, the time to find the correct $\evec$ is now given by $\max \{ |L_{1,1}|, |L_{1,1}|^2/2^{\RepMMT}, |L_{1,1}|^4/2^{\ell+\RepMMT} \}$. This is the maximum between (I) the size of starting lists, (II) the size of the output after the first merge on $\log \RepMMT$ coordinates, and (III) the size of the final output after merging on the remaining $\ell-\log \RepMMT$ coordinates. Optimization for $p, \ell$ reveals that (II) is the maximum in case of classical MMT. Overall, the expected complexity of the algorithm is

\begin{equation} \label{eq:MMTClassical}
\setlength{\abovedisplayskip}{5pt}
\setlength{\belowdisplayskip}{5pt}
T_{\mathtt{MMT}} = \P(p,l) \cdot \frac{ |L_{i,j}|^2 }{\RepMMT}.
\end{equation}

Becker-Jeux-May-Meurer in \cite{BJMM12} further improves the merging step (i.e., the dominant one) noticing that zero-coordinates of $\evec$ can be split in $\evec_1, \evec_2$ not only as $0 + 0$, but also as $1 + 1$. It turns out that constructing longer starting lists $\L_{i, j}$ using $\evec_i$ of weights $\wt(\evec_i)= p/2+\epsilon$ is profitable as it significantly increases the number of representations from $\smash{\binom{p}{p/2}}$ to $\RepBJMM:=\binom{p}{p/2} \binom{k+\ell-p}{\epsilon}$, thus allowing a better balance between the two merges: the first merge on $\log \RepBJMM$ coordinates and the second on $\ell-\log \RepBJMM$ coordinates. The expected running time of the BJMM algorithm is given by
\begin{equation} \label{eq:BJMMClassical}
T_{\mathtt{BJMM}} = \P(p,l) \cdot \frac{ |L_{i,j}|^2 }{\RepBJMM}, \quad \text{where} \; |L_{i,j}| = \binom{(k+\ell)/2}{p/4 + \epsilon}.
\end{equation}

In fact, the actual BJMM algorithm is slightly more complicated than we have described, but the main contribution comes from adding `1+1' to representations, so hereafter we refer to this simplified version as BJMM.

\subsection{Quantum ISD algorithms}

\textbf{Quantum ISD using Grover's algorithm}.
To speed-up Prange's algorithm, Bernstein in \cite{Bern10} uses Grover's search over the space of permutations, which is of size $\P(0) = {\binom{w}{k}}/\binom{n}{w}$. This drops the expected runtime from $2^{0.1206n }$ (classical) down to $2^{0.06035n}$ (quantum), cf.\ Table~\ref{table:RunTimes}.
The approach has an advantage over all the quantum algorithms we discuss later as it requires quantum registers to store data of only $\poly(n)$ size.

To obtain a quantum speed-up for other ISD algorithms like Stern's, MMT, BJMM, we need to describe quantum walks. 

\textbf{Quantum walks.} At the heart of the above ISD algorithms (except Prange's) is the search for vectors from given lists that satisfy a certain relation. This task can be generalized to the $k$-list matching problem.
\begin{definition}[$k$-list matching problem] \label{def:kListMatch}
	Let $k$ be fixed. Given $k$ equal sized lists $\L_1, \ldots, \L_k$ of binary vectors and a function $g$ that decides whether a $k$-tuple $(\vvec_1, \ldots, \vvec_k) \in \L_1 \times \ldots \times \L_k$ forms a `match' or not (outputs 1 in case of a `match'), find all $k$-tuples $(\vvec_1, \ldots, \vvec_k)  \in \L_1 \times \ldots \times \L_k$ s.t.\ $g(\vvec_1, \ldots, \vvec_k)=1$. 
\end{definition}

For example, the Stern's algorithm uses $k=2$ and its $g$ decides for a `match' whenever a pair $(\vvec_1, \vvec_2) \in \L_1 \times \L_2$ is equal on certain fixed $\ell$ coordinates. For MMT or BJMM, we deal with four lists $\L_1, \ldots, \L_4$, and function $g$ decides for the match if $\vvec_1+\vvec_2, \vvec_3+\vvec_4$ are equal on a certain part of coordinates (merge of $\L_{1}$ with $\L_2$, and $\L_3$ with $\L_4$) and, in addition, $\vvec_1+\vvec_2 + \vvec_3+\vvec_4$ is  0 on $\ell$. 


Quantumly we solve the above problem with the algorithm of Ambainis \cite{Amb04}. Originally it was described only for the case $k=2$ (search version of the so-called \emph{Element distinctness} problem), but later extended to a more general setting, \cite{ChildsEisenberg05}. We note that the complexity analysis in \cite{ChildsEisenberg05} is done in terms of \emph{query} calls to the $g$ function, while here we take into account the actual time to compute $g$.

Ambainis algorithm is best described as a quantum walk on the so-called Johnson Graph. 

\begin{definition}[Johnson graph and its eigenvalue gap]
	The Johnson graph $J(N, r)$ for an $N$-size list is an undirected graph with vertices labelled by all $r$-size subsets of the list, and with an edge between two vertices $S,S'$ iff $|S \cap S'| = r-1$. It follows that $J(N, r)$ has $\binom{N}{r}$ vertices. Its eigenvalue gap is $\delta = \frac{N}{r(N-r)}$, \cite{BCN89}.
\end{definition}


Let us briefly explain how we solve the $k$-list matching problem using quantum walks. Our description follows the so-called MNRS framework \cite{MNRS11} due to Magniez-Nayak-Roland-Santha, which measures the complexity of a quantum walk search algorithm in the costs of their Setup, Update, and Check phases. 

To setup the walk, we first prepare a uniform superposition over all $r$-size subsets $S_i \subset \L_i$ together with an auxiliary register (normalization omitted):
 \begin{align*} \label{state:QuantumWalkState}
 \sum_{S_i \subset \L_i, \; |S_i| = r }  \mkern-30mu \ket{S_1} \otimes  \ldots \otimes \ket{S_k} \otimes \ket{Aux}.
 \end{align*}
 
The auxiliary register $\ket{Aux}$ contains all the information needed to decide whether $S_1, \ldots, S_k$ contains a match. In the ISD setting, $\ket{Aux}$ stores intermediate and output lists of the matching process. For example, in Stern's algorithm ($k=2$) $\ket{Aux}$ contains all pairs $(\vvec_1, \vvec_2) \in S_1 \times S_2$ that match on $\ell$ coordinates. In case the merge is done is several steps like in MTT ($k=4$), the intermediate sublists are also stored in $\ket{Aux}$ (see Figure~\ref{fig:SternAndMMT}). 

The running time and the space complexity of the Setup phase are essentially the running time and the space complexity of the corresponding ISD algorithm with the input lists of size $|S_i| = r $ instead of $|\L_i|$. 
By the end of the Setup phase, we have a superposition over all $r$-sublists $S_1, \ldots, S_k$ of $\L_1, \ldots, \L_k$, where each $(S_1, \ldots, S_k)$ is entangled with the register $\ket{Aux}$ that contains the result of merging $(S_1, \ldots, S_k)$ into $\Sout \subset \Lout$. Also, during the creation of $\Sout$ we can already tell if it contains the error-vector $\evec$ that solves the ISD problem. When we talk about quantum space of an ISD algorithm (e.g., Table~\ref{table:RunTimes}), we mean the size of the $\ket{Aux}$ register.

Next, in the Update phase we choose a sublist $S_i$ and replace one element $\vvec_i \in S_i$ by $\vvec_i' \notin S_i$. This is one step of a walk on the Johnson graph. We update the data stored in $\ket{Aux}$: remove all the pairs in the merged lists that involve $\vvec_i$ and create possibly new matches with $\vvec_i'$. We assume the sub-lists $S_i$'s are kept sorted and stored in a data-structure that allows fast insertions/removals (e.g., radix trees as proposed in \cite{BJLM13}). We also assume that elements in $S_1, \ldots, S_k$ that result in a match, store  pointers to their match. For example, if in Stern's algorithm $\vvec_1 \in S_1, \vvec_2 \in S_2$ give a match, we keep a pointer to  $\vvec_1+\vvec_2 \in \Sout$ and also a pointer from $\vvec_1+\vvec_2$ to $\vvec_1 \in S_1, \vvec_2 \in S_2$.

After we have performed $\Theta(1/\sqrt{\delta})$ updates (recall, $\delta$ is the eigenvaule gap of $J(N, r)$), we check if the updated register $\ket{S_1} \otimes  \ldots \otimes \ket{S_k} \otimes \ket{Aux}$  gives a match. This is the Checking phase.

Thanks to the MNRS framework, once we know the costs of (a) the Setup phase $T_{\mathtt{S}}$, (b) the Update phase $T_{\mathtt{U}}$, and (c) the Checking phase $T_{\mathtt{C}}$, we know that after $T_{\mathtt{QW}}$ many steps, we measure a register $\ket{S_1} \otimes  \ldots \otimes \ket{S_k} \otimes \ket{Aux}$ that contains the correct error-vector with overwhelming probability, where
\begin{equation}\label{eq:QWRuntime}
\setlength{\abovedisplayskip}{3pt}
\setlength{\belowdisplayskip}{3pt}
T_{\mathtt{QW}} = T_{\mathtt{S}} + \frac{1}{\sqrt{\eps}} \left(\frac{1}{\sqrt{\delta}} \cdot T_{\mathtt{U}} + T_{\mathtt{C}}\right).
\end{equation}

In the above formula, $\eps$ is a fraction of vertices in $J(N,r)$ that contain the correct error-vector. For a fixed $k$, we have $\eps \approx r^k / N^k$ where $N = |\L_1| = \ldots =|\L_k|$. 
Strictly speaking, the walk we have just described is a walk on a $k$-Cartesian product of Johnson graphs -- one for each sublist $S_i$, so the value $\delta$ in Eq.~\eqref{eq:QWRuntime} must be the eigenvalue gap for such a large graph. As proved in \cite[Theorem 2]{KT17}, for fixed constant $k$, it is lower-bounded by $\frac{N}{k \cdot r (N-r)}$. The analysis of \cite{KT17} as well ours are asymptotical, so we ignore the constant factor of $1/k$. 
An optimal choice for $r$ that minimizes Eq.~\eqref{eq:QWRuntime} is discussed in the next section.



\noindent \textbf{Kachigar-Tillich quantum ISD algorithms.} 
The quantum walk search algorithm described above solves the ISD problem provided we have found a permutation $\pi$ that gives the desired distribution of 1's in the error-vector. Kachigar and Tillich in \cite{KT17} suggest to run Grover's algorithm for $\pi$ with the `checking' function for Grover's search being a routine for the $k$-list matching problem. Their ISD algorithm performs transformations on the quantum state of the form (normalization omitted):
\begin{align} \label{eq:MetaQuantumState}
\sum_{i=1}^{\P(p, \ell)} \ket{ \pi_i} \ket{\pi_i(\HMat)} \otimes \mkern-20mu
\underbrace{\Bigg[ \sum_{\mathclap{\substack{\hspace{3.5em} S_i \subset \L_i, \; |S_i| = r }}}  \ket{S_1} \otimes  \ldots \otimes \ket{S_k} \otimes \ket{Aux} \Bigg]}_{\text{Quantum Walk = Check for the outer Grover}} \mkern-20mu
\otimes \ket{ \text{Is } \pi \text{ good?} }
\end{align}
The outer-search is Grover's algorithm over $\P(p,\ell)$ permutations, where  $\P(p,\ell)$ is chosen such that we expect to have one $\pi$ that leads to a good permutation of 1's in the error-vector (see Eq.~\eqref{eq:PermutationProbabiltiyFS}). The check if a permutation $\pi$ is good is realized via quantum walk search for $k$ vectors $\vvec_1, \ldots, \vvec_k \in S_1 \times \ldots \times S_k$ that match on certain coordinates and lead to the correct error vector. Note an important difference between classical and quantum settings: during the quantum walk we search over sublists $S_i \subset \L_i$ which are exponentially shorter than $\L_i$.

After $T_{\mathtt{QW}}$ steps, the register $\ket{Aux}$ contains a $k$-tuple $(\vvec_1, \ldots, \vvec_k)$ that leads to the correct error vector provided a permutation $\pi$ is good. Hence, after $\softO(\sqrt{\P(p, \ell)} \cdot T_{\mathtt{QW}})$ steps, the measurement of the first register gives a good $\pi$ with constant success probability. The resulting state will be entangled with registers that store $S_1, \ldots, S_k$ together with the pointers to the matching elements. Once we measure $S_1, \ldots, S_k$, we retrieve these pointers and, finally, reconstruct the error vector as in the classical case.

\noindent \textbf{Quantum Shamir-Schroeppel technique} was introduced in \cite{ShamirSchro81} to reduce the memory complexity of a generic meet-in-the-middle attack, i.e., the $k$-list matching problem for $k=2$. Assume we want to find a pair $\vvec_1, \vvec_2 \in \L_1 \times \L_2$ s.t.\ $\vvec_1 = \vvec_2$ on certain $\ell$ coordinates. Assume further that we can decompose $\L_1 = \L_{1,1} + \L_{1,2}$ s.t. $|\L_{1,1}| = |\L_{1,2}| = \sqrt{|\L_1|}$ (analogously, for $\L_2$). 
The idea of Shamir and Schroeppel is to guess that the correct vectors $\vvec_1, \vvec_2$ are equal to some $\tvec \in \GF^{\ell'}$ on $\ell' \leq \ell$ coordinates and enumerate all such pairs. Namely, we enumerate $\vvec_1$ by constructing $\L_1$ in the meet-in-the-middle way from $\L_{1,1}, \L_{1,2}$ in time $\max\{\sqrt{|\L_1|},|\L_1|/2^{\ell'} \}$, s.t.\ $\L_1$ only contains vectors that are equal to $\tvec$ on $\ell'$ (same for $\L_2$). Classically, we make $2^{\ell'}$ guesses for $\tvec$, so the overall time of the algorithm will be $|\L_1|$ (same as naive $2$-list matching), but we save in memory. 
 
In \cite{KT17}, in order to improve not only in memory, but also in time, the authors run Grover's search over $2^{\ell'}$ guesses for $\tvec$. Indeed, this gives a speed-up for ISD algorithms that solve the $2$-list matching problem (cf.\ the complexities of Dumer's algorithm in Table~\ref{table:RunTimes}).

\section{Quantum MMT and BJMM algorithms}\label{sec:2ListsQuantum}

In this section we analyse the complexity of quantum versions of
MMT and BJMM ISD algorithms given in \cite{KT17}. We note that the way we apply and analyse quantum walks to ISD closely resembles Bernstein's et al. algorithm for Subset Sum \cite{BJLM13}.

Let us first look at the generalized version of a quantum ISD algorithm, where we can plug-in any of the ISD algorithms described in Sect.~\ref{sec:Prelims}. Recall that on input we receive $(\HMat, \svec) \in \GF^{(n-k) \times n} \times \GF^{n-k}$, and are asked to output $\evec \in \GF^n$ of weight $\wt(\evec) = w$ that satisfies $\HMat \evec\transpose = \svec$. Alg.~\ref{alg:MetaQuantumISD} below can be viewed as a `meta' quantum algorithm for ISD.
\vspace{-10pt}
\begin{algorithm}[H]
	\caption{A quantum ISD algorithm}
	\label{alg:MetaQuantumISD}
	\vspace{5pt}
	\begin{algorithmic}[1]
		\State Prepare a superposition over $\P$-many permutations $\pi$
		\State For each $\pi$
			\begin{algsubstates}
					\State Setup a superposition $\ket{S_1} \otimes  \ldots \otimes \ket{S_k} \otimes \ket{Aux}$ for $S_i \subset \L_i$, $|S_i| = r$ 
					\State Run a quantum walk search on $\ket{S_1} \otimes  \ldots \otimes \ket{S_k} \otimes \ket{Aux}$ to find a matching tuple $(\evec_1, \ldots, \evec_k) \in S_1 \times \ldots \times S_k$, if exists; indicate otherwise that no tuple is found.
					\label{alg_line:Walk}
			\end{algsubstates}
		\State Apply amplitude amplification (Grover's search)  on Step 1 for those $\pi$ that led to a match on Step 2.b. Measure the register $\pi$ and then the register $\ket{Aux}$. 
	\end{algorithmic}
\end{algorithm}
\vspace{-10pt}

The algorithm is parametrized by (I.) the size of the permutation space $\P$ we iterate over in order to find the desired distribution of 1's in the solution (e.g., Eq.~\eqref{eq:PermutationProbabiltiyFS} for MMT); (II.) $k$ -- the number of staring lists $\L_i$'s an ISD-algorithm considers (e.g., $k=0$ for Prange, $k=2$ for Stern/Dumer, $k=4$ for MMT);  (III.) $r$ -- the size of $S_i$'s, $1 \leq i \leq k$. The asymptotic complexity of Alg.~\ref{alg:MetaQuantumISD} will depend on these quantities as we now explain in detail. 

 
Step 1 consists in preparing a superposition $\smash{\sum_{i=1}^{\P} \ket{ \pi_i} \ket{\pi_i(\HMat)} }$, which is efficient. 
Step 2 is a quantum walk algorithm for the $k$-list matching problem, i.e. search for all $(\evec_1, \ldots, \evec_k) \in \L_1 \times \ldots \times \L_k$ from which the solution vector can be constructed. The cost of Step 2 can be split into the cost of the Setup phase (Step 2(a)) and the cost of the Update and Check phases (Step 2(b)). 

The cost of the Step 2(a)  -- preparing a superposition over $k$-tensor product of $S_i \subset \L_i$ and computing the data for $\ket{Aux}$  -- is essentially the cost of a classical ISD algorithm, where on input instead of the lists $\L_i$'s, we consider sublists $S_i$ of size $r \ll  |\L_i|$. Recall that `computing the data for $\ket{Aux}$' means constructing the subset $\Sout \subset \Lout$ using only elements from $S_i$'s (see Fig.~\ref{fig:SternAndMMT}). Step 2(b) performs a quantum walk over the $k$-Cartesian product of Johnson Graphs, $J(|\L_i|, r)^{\otimes k}$, with eigenvalue gap $\delta = \Theta({|\L_i|}/( r \cdot (|\L_i|-r) ) ) \approx 1/r$ for $r \ll |L_i|$. To estimate $\eps$ -- the fraction of $(S_1, \ldots, S_k)$ that give the solution, note that with probability $\Theta(r/|L_i|)$, an $r$-size subset $S_i$ contains an element $\evec_i$ that contributes to the solution. Hence, $k$ such subsets -- one vertex of $J(|\L_i|, r)^{\otimes k}$ -- contain the solution $(\evec_1, \ldots, \evec_k)$ with probability $\smash{\eps = \left( {r}/{|\L_i|}\right)^k}$. 

Now we focus on the Update and Check phases. Recall that at these steps we replace one element from a list $\L_i$ and, to keep the state consistent, remove the data from $\ket{Aux}$ that was generated using the removed elements, and compute the data in $\ket{Aux}$ for the newly added element. Hence, asymptotically the expected cost is the number of elements we need to recompute in the lists contained in $\ket{Aux}$ (for example, for MMT or BJMM algorithms, it is the number of elements in $\L_1, \L_2, \Lout$ affected by the replacement of one element in a starting list $L_i$). 
Once we know the time to create $\ket{Aux}$, $\delta$ and $\eps$, we obtain the total complexity of Step~2 from Eq.~\eqref{eq:QWRuntime}.

Finally, Grover's search over $\P$-many permutations requires $\sqrt{\P}$ calls to a `checking' function for a measurement to output a good $\pi$. The measurement will collapse the state given in Eq.~\eqref{eq:MetaQuantumState} into a superposition of $\ket{S_1} \otimes  \ldots \otimes \ket{S_k} \otimes \ket{Aux}$, where the amplitude of those $\ket{Aux}$ that contain the actual solution $\evec$ will be amplified. Measurement of $\ket{Aux}$ leads to the solution.
Regarding Step 2 as `checking' routine for amplitude-amplification of Step 1 and assuming that an ISD algorithm on input-lists of size $r$ has classical running time $T_{\mathtt{ISD}}(r)$, we obtain the following complexity of Alg.~\ref{alg:MetaQuantumISD}:
\begin{theorem}\label{thm:MetaAlgRunTime}
	Assume we run Alg.~\ref{alg:MetaQuantumISD} with a classical ISD algorithm with the following properties: (I.) it expects after $\P$ permutations of the columns of $\HMat$ to find the desired weight-distribution for the error-vector, (II.) it performs the search for the error-vector over $k$ lists each of size $|\L|$ in (quantum)  time $T_{\mathtt{ISD}}(|\L|)$, (III.) replacing one element in any of these $k$ lists, costs $T_{\mathtt{U}}$ to update the $\ket{Aux}$ register at Step 2(b). Then for $r \ll |\L|$ satisfying $\smash[below]{T_{\mathtt{ISD}}(r) = \softO\Big(\sqrt{\frac{|\L|^k}{r^{k-1}} } \cdot T_{\mathtt{U}} \Big)}$, the expected running time of Alg.~\ref{alg:MetaQuantumISD} is
	\begin{equation*}
	\setlength{\abovedisplayskip}{8pt}
	\setlength{\belowdisplayskip}{8pt}
		T_{\mathtt{ISD_Q}}= \sqrt{\P} \cdot  T_{\mathtt{ISD}} (r) .
	\end{equation*}
	In particular, the MMT algorithm has $k=4$, $\P = \P(p, \ell)$ given in Eq.~\eqref{eq:PermutationProbabiltiyFS}, $|\L|$ given in Eq.~\eqref{eq:MMTListSize} and $\RepMMT:=\binom{p}{p/2}$. Under the (heuristic) assumption that elements in all lists are uniform and independent from $\{0,1\}^n$, we expect $T_{\mathtt{U}}=\frac{r}{\RepMMT}$, leading to
	\begin{equation*}
	\setlength{\abovedisplayskip}{5pt}
	\setlength{\belowdisplayskip}{5pt}
		T_{\mathtt{MMT_Q}} = \softO \Big( \sqrt{ \P(p, \ell)} \cdot \frac{|\L|^{\frac{6}{5}}}{\sqrt{\RepMMT }} \Big).
	\end{equation*}
	Similarly, for the BJMM algorithm \cite{BJMM12} with starting lists-sizes $|\L|$ given in Eq.~\eqref{eq:BJMMClassical}, expected $T_{\mathtt{U}}$ is $\frac{r}{\RepBJMM}$, and $\RepBJMM:=\binom{p}{p/2} \binom{k+\ell-p}{\epsilon}$ for some $\epsilon \geq 0$, we have
	\begin{equation*}
	T_{\mathtt{BJMM_Q}} = \softO \Big( \sqrt{ \P(p, \ell)} \cdot \frac{|\L|^{\frac{6}{5}}}{\sqrt{\RepBJMM }} \Big).
	\end{equation*}
\end{theorem}

\begin{proof}
	The first statement follows from the discussion before the theorem: Grover's search for a good $\pi$ makes $\sqrt{\P}$ `calls', where each `call' is a quantum walk search of complexity $T_{\mathtt{Setup}} + \frac{1}{\sqrt{\eps}} \cdot  (\frac{1}{\sqrt{\delta}}  T_{\mathtt{U}}+T_{\mathtt{C}})$. The condition on $r$ is set such that the Steps 2(a) and 2(b) in Alg.~\ref{alg:MetaQuantumISD} are asymptotically balanced, namely, we want $T_{\mathtt{Setup}} = \smash{ \frac{1}{\sqrt{\eps}} \cdot \frac{1}{\sqrt{\delta}} } \cdot T_{\mathtt{U}}$, cf.\ Eq.~\eqref{eq:QWRuntime}. 
	 We have $T_{\mathtt{Setup}} = T_{\mathtt{ISD}}(r) $, $\delta \approx 1/r$, $\eps = \left( {r}/{|\L_i|}\right)^k$, the cost of one update step is $T_{\mathtt{U}}$, and the checking phase is $\softO(\log r)$ (as it consist in checking if $\ket{Aux}$ contains the solution, which can be done in time $\softO(\log|\Lout|)$ when $\Lout$ is kept sorted). With this, the optimal choice for $r$ should satisfy $T_{\mathtt{ISD}}(r) = \softO \big(\sqrt{{|\L|^k}/{r^{k-1}}} \cdot T_{\mathtt{U}} \big).$
	
	For the classical MMT algorithm, the dominating step is the construction of the lists $\L_1, \L_2$ whose elements are already equal on a certain number of coordinates denoted $\log \RepMMT \approx p$ in Sect.~\ref{sec:Prelims}  (see `middle' lists in Fig.~\ref{fig:SternAndMMT} in the right figure). Quantumly, however, Kachigar and Tillich in \cite{KT17} observed that if we assume instead that the dominant step is creation of the lists $L_{i,j}$ (the `upper' lists in Fig.~\ref{fig:SternAndMMT}), we obtain a slightly faster algorithm. The reason is in Shamir-Schroeppel technique: we construct the list $\L_1, \L_2$ by 1. forcing elements in $\L_1$ and $\L_2$ be equal to a vector $\tvec \in  \GF^{\ell'}$ on $\ell'< n$ coordinates, 2. looping over all possible $2^{\ell'}$ vectors $\tvec$. Quantumly, the loop costs $2^{\ell'/2}$ `calls' (again, here a `call' is a quantum walk).
	Hence, taking creation the lists $\L_{i, j}$ as the dominant one, the setup phase (or Step 2(a) of Alg.~\ref{alg:MetaQuantumISD}) is of complexity $r \cdot 2^{\ell'/2}$, where $r$ is the size of the sets $S_{i,j} \subset \L_{i,j}$. Both parameters $\ell', r$ are subject to optimizations.
	
	To determine $T_{\mathtt{U}}$, we remind that $|S_{i,j}|=r$ and $|S_i| = \frac{|S_{i,j}|^2}{2^{\ell'} \cdot \RepMMT}$, where $S_i$ is obtained by considering all pairs $(\vvec, \vvec') \in S_{1,1} \times S_{1,2}$ that are equal to a fixed vector $\tvec$ on $\ell'$-coordinates and are equal to another fixed value on $\log \RepMMT$- coordinates.
	Under the assumption that all elements are uniform random and independent, changing one element in $S_{i,j}$, would require to recompute $\frac{|S_i|}{2^{\ell'} \cdot \RepMMT} $ - many elements in $S_i$. Similarly, changing one element in $S_i$ leads to changing $\frac{|\Sout|}{2^{n-\ell'-\log \RepMMT}}$-many elements in $\Sout$. To simplify the analysis, we choose the parameters $r, \ell', \RepMMT$ such that $\frac{|S_i|}{2^{\ell'} \cdot \RepMMT} \approx \frac{|\Sout|}{2^{n-\ell'-\log \RepMMT}} = \softO(\log r)$, that is, $T_{\mathtt{U}}$ is irrelevant asymptotically. This simplification puts the following constraint on the parameters $r \overset{!}{=} \ell'+\log \RepMMT$.
	 
	The analysis now simplifies to the balancing constraint between the setup phase for the quantum walk (which the creation of the sets $S_{i,j}$ of size $r$) and $\frac{1}{\sqrt{\eps \cdot \delta}}$. Solving $r \overset{!}{=} \smash{ \sqrt{\frac{|\L|^k}{r^{k-1}}} }$ for $r$, we receive $\smash{r = |\L|^{\frac{4}{5}}}$ as the optimal size for  $S_i$'s. Hence, the running time of the Step 2 of Alg.~\ref{alg:MetaQuantumISD} for MMT is  $\softO( {|\L|^{\frac{4}{5}}} \cdot 2^{\ell'/2}$, where the second multiple is Grover's itearion over $2^{\ell'}$ vectors $\tvec$. From the constraint $r \overset{!}{=} \ell'+\log \RepMMT$, we obtain $ 2^{\ell'} = |\L|^{\frac{4}{5}} / \RepMMT$, and hence, the second statement of the teorem.
	
	The BJMM algorithm differs from MMT in the number of representations $\RepBJMM $  and the size of the starting lists $|\L|$. Similar to MMT, we choose  $r=|\L|^{\frac{4}{5}}$, the complexity of the quantum walk for BJMM becomes $\softO( {|\L|^{\frac{4}{5}}} / \sqrt{\RepBJMM})$. $\hfill \qed$
\end{proof}


	The above complexity result gives formulas that depend on various parameters. In order to obtain the figures from Table~\ref{table:RunTimes}, we run an optimization program that finds parameters values that minimize the running time $T_{\mathtt{ISD_Q}}$ under the constraints mentioned in the above proof. While we do not prove that these value are global optima, the values we obtain are feasible (they satisfy all the constraints), and hence, can be used inside the decoding algorithm.
	 We chose use the optimization package implemented in Maple. The optimization program for  Table~\ref{table:RunTimes}  is available at \url{http://perso.ens-lyon.fr/elena.kirshanova/}.
	
	From the table one can observe that classically, the improvements over Prange achieved by recent algorithms are quite substantial: BJMM gains a factor of $2^{0.019 \cdot n}$ in the leading-order term. Quantumly, however, the improvement is less pronounced. The reason lies in the fact that the speed-up coming from Grover's search is much larger than the speed-up offered by the quantum walk. Also, the $k$-list matching problem become harder (quantumly) once we increase $k$ because the fraction of `good' subsets $\eps$ becomes smaller. 

\section{Decoding with Near Neighbour Search} \label{sec:DecodingWithNN}
For a reader familiar with Indyk-Motwani locality-sensitive hashing \cite{IM98} for  Near Neighbour search (defined below), Stern's algorithm and its improvements \cite{Dumer} essentially implement such hashing by projecting on $\ell$-coordinates and applying it to the lists $\L_1, \L_2$. In this section, we consider another Near Neighbour technique.

\subsection{Re-interpretation of May-Ozerov Near Neighbour algorithm} \label{subsec:Reinterp}

The best known classical ISD algorithm is due to May-Ozerov \cite{MO15}. It is based on the observation that ISD is a similarity search problem under the Hamming metric. In particular, Eq.~\eqref{eq:ISDeq} defines the approximate relation:
\begin{equation} \label{eq:ISDeqNN}
	\QMat \evec_1 \approx \QMat \evec_2 + \bar{\svec}.
\end{equation}
The approximation sign $\approx$ means that the Hamming distance between the left-hand side and the right-hand side of Eq.~\eqref{eq:ISDeqNN} is at most $\wt(\evec_3) = w-p$ (cf.\ Eq.~\eqref{eq:ISDeqExt}). Enumerating over all $\evec_1$ and $\evec_2$, we receive an instance of the $(w-p)$-Near Neighbour (NN) problem:

\begin{definition}[$\gamma$-Near Neighbour] \label{def:NNProblem}
	Let $\L \subset \GF^n$ be a list of uniform random binary vectors. The $\gamma$-Near Neighbour problem consists in preprocessing $\L$ s.t.\ upon receiving a query vector $\qvec \in \GF^n$, we can efficiently find all $\vvec \in \L$ that are $\gamma$-close to $\qvec$, i.e., all $\vvec$ with $\dist(\vvec, \qvec) \leq \gamma \cdot n$ for some $\gamma \leq 1/2$.\footnote{The (dimensionless) distances we consider here, denoted further $\gamma, \alpha, \beta$, are all $\leq 1/2$, since we can flip the bits of the query point and search for `close' rather than `far apart' vectors.}
\end{definition}

Thus the ISD instance given in Eq.~\eqref{eq:ISDeqNN} becomes a special case of the $(w-p)$-NN problem with $\L = \{\QMat \evec_1\}$ for all $\evec_1 \in \GF^{(k+\ell)/2} \times \zvec^{(k+\ell)/2} \times $, and the queries taken from $\{ \QMat \evec_2+\svec\}$ for all $\evec_2 \in \zvec^{(k+\ell)/2} \times \GF^{(k+\ell)/2}$. In \cite{MO15}, the algorithm is described for this special case, namely, when the number of queries is equal to $|\L|$ and all the queries are explicitly given in advance. So it is not immediately clear how to use their result in quantum setting, where we only operate on  the sublists of $\L$ and update them with new vectors during the quantum walk.

In this section, we re-phrase the May-Ozerov algorithm in terms more common to the near neighbour literature, namely in the `Update' and `Query' quantities. It allows us to use the algorithm in more general settings, e.g., when the number of queries differs from $|\L|$ and when the query-points $\qvec$ do not come all at once. This view enables us to adapt their algorithm to quantum-walk framework. 


The main ingredient of the May-Ozerov algorithm is what became known as Locality-Sensitive Filtering (LSF), see  \cite{BDGL16} for an example of this technique in the context of lattice sieving. In LSF we create a set $\C \subset \GF^n$ of \emph{filtering} vectors $\cvec$  which divide the Hamming space into (possibly overlapping) regions. These regions are defined as Hamming balls of radius $\alpha$ centred at $\cvec$, where $\alpha$ is an LSF-parameter we can choose. So each filtering vector $\cvec\in \C$ defines a region $\Region_\cvec$ as the set of all vectors that are $\alpha$-close to $\cvec$, namely, $\Region_\cvec = \{ \vvec \in \GF^n: \dist(\vvec, \cvec) \leq \alpha  \}$. Drawing an analogy with Locality-Sensitive Hashing, these filtering vectors play role of hash-functions. In LSF, instead of defining a function, we define its pre-image. 

The preprocessing for the input list $\L$ consists in creating a large enough set $\C$ of filtering vectors and assigning all $\vvec \in \L$ to their regions (see the $\Call{Insert}{\vvec}$ function in Alg.~\ref{alg:LSF} below). This assignment defines the LSF buckets as $\Bucket_\cvec = \Region_\cvec \cap \L$. The LSF data structure $\D$ consists of the union of all the buckets. In the course of quantum walk search, we will also need to remove vectors from $\D$. For that we have the $\Call{Remove}{\vvec}$ function which deletes $\vvec$ from all the buckets $\Bucket_\cvec$ containing $\vvec$. Note that for each $\Bucket_\cvec$ both $\Call{Insert}$ and $\Call{Remove}$ can be implemented in time $\softO(\log |\Bucket_\cvec|)$ if we store the buckets as, for example, binary trees. Finally, in order to answer a query $\qvec$, we look at all buckets $\Bucket_\cvec$  that are $\beta$-close to $\qvec$ (i.e., all $\cvec \in \C$ with $\dist(\qvec, \cvec) \leq \beta$), and we check if any of the vectors stored in these $\beta$-close buckets gives a solution to $\gamma$-Near Neighbour. As it is typically the case for NN-algorithms \cite{Laa16}, we have two trade-off parameters $(\alpha, \beta)$: the closer $\alpha$ to $1/2$, the more buckets we should create, but the query is fast because we may allow small $\beta$. Making $\alpha$ smaller reduces the prepocessing cost but requires more work during queries.

\begin{algorithm}[h]
	\caption{%
		Algorithms for Locality-Sensitive Filtering 
		\label{alg:LSF}
	}
	\begin{algorithmic}[1]
		\Statex Parameters: 
		\Statex \hspace{10pt} $\alpha$ - the insertion parameter 
		\Statex \hspace{10pt} $\beta$ - the query parameter 
		\Statex \hspace{10pt} $\gamma$ - the target distance
		\Statex \hspace{10pt} $\C$ - the set of filtering vectors
		\Statex \hspace{10pt} $\D$ - the LSF data structure: $\D = \cup_{\cvec \in \C} \Bucket_\cvec$

		\vspace{6pt}
		
		\Function{Insert}{$\xvec$}\Comment{Add $\xvec$ to all the relevant buckets of $\D$}
		\ForAll{$\cvec \in \C$ s.t.\ $\dist(\cvec, \xvec) \leq \alpha$}
		\State $\Bucket_\cvec \gets \Bucket_\cvec \cup \{\xvec \}$
		\EndFor
		\EndFunction
	\end{algorithmic}
	
	\vspace{3pt plus 3pt}
	\begin{algorithmic}[1]
		\Function{Remove}{$\xvec$} \Comment{Remove $\xvec$ from all buckets}
		\ForAll{$\cvec \in \C$ s.t.\ $\dist(\cvec, \xvec) \leq \alpha$}
		\State $\Bucket_\cvec \gets \Bucket_\cvec \setminus \{\xvec\}$
		\EndFor
		\EndFunction
	\end{algorithmic}
	
	\vspace{3pt plus 3pt}
	\begin{algorithmic}[1]
		\Function{Query}{$\qvec$} \Comment{Find all $\xvec \in \D$ with $\dist(\xvec,\qvec) \leq \beta$}
		\State $\textrm{CloseVectors} \gets \emptyset$
		\ForAll{$\cvec \in \C$ s.t.\ $\dist(\cvec,\qvec) \leq \beta$}
			\ForAll{$\xvec \in \Bucket_\cvec$}
				\If{$\dist(\xvec,\qvec) \leq \gamma$}
					\State $\textrm{CloseVectors}\gets \textrm{CloseVectors} \cup \smash{\{\xvec\}}$
				\EndIf
			\EndFor
		\EndFor
		\State\Return $\textrm{CloseVectors}$
		\EndFunction
		
	\end{algorithmic}
\end{algorithm}

\textbf{Structured filter-vectors or the `strips technique'}. In the main procedures of LSF, $\Call{Update}{}, \Call{Remove}{}$, and $\Call{Query}{} $, we are required to find all close buckets for a given point. Naive search finds these buckets time $|\C|$ which is inefficient. We can do better by making filter-vectors $\cvec$ structured. The technique has several names, `strips' in \cite{MO15}, `random product code' in \cite{BDGL16}, and `tensoring' in \cite{Chr17}, but either way it amounts to the following. Each vector $\cvec$ is a concatenation of several codewords from some low-dimensional codes (so, $\C $ is a Cartesian product of all these codes). All $\cvec$'s close to $\xvec$ are obtained by iteratively decoding the relevant projections of $\xvec$ under the codes defined on these projections  (say, for $\xvec=[\xvec_1 || \ldots || \xvec_\ell]$, we start by decoding $\xvec_1$). On each iteration, we filter out those $\cvec$'s that are guaranteed to be far from $\xvec$ (i.e, only $\cvec=[\cvec_1 || \ldots || \cvec_\ell]$'s with $\cvec_1$ close to $\xvec_1$ are kept). Choosing the lengths of low-dimensional codes carefully enough, we can ensure that $\cvec$'s are sufficiently close to independent random vectors. This trick allows us find all close buckets in time (up to lower-order terms) equal to the output size. We refer the reader to \cite{BDGL16,MO15} for details.

Before  we give complexities for the routines described in Alg.~\ref{alg:LSF} as functions of $\alpha, \beta, \gamma$, we recall the definition of the entropy function for a discrete probability distribution defined by a vector $\pvec$.
\begin{definition}
	Let $\pvec \in \R^t$ be a real vector that represents a certain probability distribution, i.e., $\pvec$ satisfies $0 \leq \pvec_i \leq 1,$ and $ \sum_{i=1}^{t} \pvec_i = 1$.
	Then $H(\pvec) $ is the entropy of the distribution $\pvec$:
	\[
		H(\pvec) = - \sum_{i} \pvec_i \log \pvec_i.
	\]
\end{definition} 
	
	We will be using the above definition in the following context: 	
	Let $(\xvec_1, \ldots, \xvec_m)$ be an $m$-tuple of vectors from $\GF^n$ and let $\pvec \in \R^{2^m}$ be a real vector indexed by all $m$-length binary vectors that represents the distribution of the $m$-tuple. That is, $\pvec_{i_1 \ldots i_m}$ counts the number of occurrences (relative to $n$) of the coordinates' configuration:  $\pvec_{i_1 \ldots i_m} = |\{c : \xvec_1[c] = i_1, \ldots, \xvec_m[c]=i_m\}|$. Such $\pvec$ defines a discrete probability distribution on $\{1, \ldots, 2^m\}$.
	
	For example, consider a random $2$-tuple $(\xvec_1, \xvec_2)$ with $\dist(\xvec_1, \xvec_2) = w$. Its distribution vector is $\pvec =(p_{00}, p_{01}, p_{10}, p_{11})$ satisfying $p_{01}+p_{10} = w$ and $p_{00} = | \{c: \xvec_1[c] = 0, \xvec_2[c] = 0\} |$,  $p_{01} = | \{c: \xvec_1[c] = 0, \xvec_2[c] = 1\} |$, $p_{10} = | \{c: \xvec_1[c] = 1, \xvec_2[c] = 0\} |$, $p_{11} = | \{c: \xvec_1[c] = 1, \xvec_2[c] = 1\} |$. 
	
	In case $\xvec_1$ is fixed, we can shift the tuple: $(\zvec, \xvec_2-\xvec_1)$, and obtain $\pvec=(1-w, w)$ with $H(\pvec) = H(w) =  -w \log w - (1-w) \log(1-w),$ which just counts the number of all binary vectors of weight $w$.
	

In the following, we give complexities of the Near Neighbour problem routines.
We assume throughout that the target distance $0\leq \gamma \leq 1/2$ and the parameters $0 \leq \alpha, \beta \leq 1/2$ are fixed. We remark that the choice of $\pvec$ given in the next lemma corresponds to the so-called \emph{balanced configuration} used for lattice-sieving in \cite{HK17,HKL}. A \emph{configuration} describes certain pairwise properties between $k$-tuples of vectors:  for lattice-sieving the interesting property is the inner-product between each pair of vectors in a $k$-tuple, while in our case, it is the Hamming distance we are concerned about. The lemma below describes a configuration $\pvec$ attained by almost all triples $(\vvec, \xvec, \qvec)$) with prescribed pairwise Hamming distances.
\begin{lemma}[Size of $\C$]
	To answer a Near Neighbour query $\qvec$ with the $\Call{Query}{\qvec}$ procedure from Alg.~\ref{alg:LSF}, i.e., output all $\vvec \in \L$ s.t.\ $\dist(\qvec, \vvec) \leq \gamma$ with super-exponentially small error \footnote{By `error' we mean missing a vector which is $\gamma$-close to $\qvec$.}, the total number of buckets $\C$ should be (up to sub-exponential factors)
	\begin{equation} \label{eq:NumOfBuckets}
		|\C| = 2^{(1-(H(\pvec(\alpha, \beta, \gamma)) - H(\pvec(\gamma))) \cdot n},
	\end{equation} 
	where $\pvec(\alpha, \beta, \gamma) \in \R^8$ and $\pvec(\gamma) \in \R^4$ are probability distributions that satisfy

	\begin{tabular}{L L} 
		\[ \pvec(\alpha, \beta, \gamma) :  \begin{cases}  p_{000}=p_{111} = \frac{1}{2} - \frac{1}{4}(\gamma+\beta+\alpha) \\ 
						 p_{001}=p_{110} = \frac{1}{4} (\gamma+\beta-\alpha) \\
						 p_{010}=p_{101} = \frac{1}{4} (\gamma+\alpha-\beta)\\
						 p_{100}=p_{011} = \frac{1}{4} (\beta+\alpha -\gamma)\\
		 \end{cases} \] &
		\[ \pvec(\gamma) : \begin{cases} p_{00}=p_{11} = \frac{1-\gamma}{2} \\
						 p_{10}=p_{01} = \frac{\gamma}{2}.
		\end{cases} \]   \tabularnewline
	\end{tabular}
\end{lemma}

\begin{proof}
	
	Consider a pair $(\vvec, \qvec)$  s.t.\ $\dist(\vvec, \qvec) = \gamma$. 
	The number of filtering vectors $|\C|$ is determined by the inverse of the probability that a random vector $\cvec$ will `find' this pair, namely
		\begin{equation} \label{eq:SizeOfC_proof}
			|\C| = 1/ \Pr_{\cvec \in \{0,1\}^n}[\dist(\cvec, \vvec) = \alpha, \dist(\cvec, \qvec) = \beta \; |\; \dist(\vvec, \qvec) = \gamma].
		\end{equation}
Note that we switched from the `$\leq$' sign to the `$=$' sign for distances. For $\alpha, \beta$ cases, this is a legitimate change, since
		 \begin{align*} \label{eq:max_prob_in_thm}
			\sum_{\alpha'=0}^{\alpha} \sum_{\beta'=0}^{\beta} \Pr_{\cvec \in \{0,1\}^n}[\dist(\cvec, \vvec) = \alpha', \dist(\cvec, \qvec) = \beta' \; |\; \dist(\vvec, \qvec) = \gamma] \\
			\approx \Pr_{\cvec \in \{0,1\}^n}[\dist(\cvec, \vvec) = \alpha, \dist(\cvec, \qvec) = \beta \; |\; \dist(\vvec, \qvec) = \gamma].
		 \end{align*}
The approximate equality holds as each summand is exponential in $n$ and $\alpha, \beta \leq 1/2$. So the sum attains its maximum at largest $\alpha, \beta$ (otherwise, we could have decreased $\alpha$ and/or $\beta$, which will not affect the success probability but will make the search faster). So up to $\poly(n)$ factors, the above sum is determined by the maximal summand. At the end of the proof we argue on the validity of the sign change for $\gamma$.

If we want to find all but super-exponentially small fraction of $\vvec$'s for a given $\qvec$, we increase $|\C|$ by a $\poly(n)$ factor for some large enough polynomial and obtain the result by Chernoff bounds.  
	
The denominator of Eq.~\eqref{eq:SizeOfC_proof} is (assuming $\cvec, \qvec$, and $\vvec$ are uniform)
	\begin{equation} \label{eq:InverseSizeOfC_proof}
		\frac{\Pr\limits_{\cvec, \vvec, \qvec} [\dist(\cvec, \vvec) = \alpha, \dist(\cvec, \qvec) = \beta,  \dist(\vvec, \qvec) = \gamma]}{\Pr\limits_{\vvec,\qvec}[\dist(\vvec, \qvec) = \gamma]} = \frac{ 2^{H(\pvec(\alpha, \beta, \gamma)) \cdot n} / 2^{3n} }{2^{H(\pvec(\gamma)) \cdot n} / 2^{2n} } = \frac{1}{|\C|},
	\end{equation}
	for some distribution vectors $\pvec(\alpha, \beta, \gamma), \pvec(\gamma)$.
	
The  statement about the entries of the vector $\pvec(\alpha, \beta, \gamma)$ comes from the following three facts (the entries for $\pvec(\gamma)$ are straightforward to obtain):
		\begin{itemize}[topsep=0pt]
			\renewcommand\labelitemi{--}
			\setlength\itemsep{0em}
			\item the distance constraints: three for $\pvec(\alpha, \beta, \gamma)$ and one for $\pvec(\gamma)$,
			\item the uniformity of $\vvec, \cvec$ and $\qvec$ (this allows to assume that the contribution of $p_{01}=p_{010}+p_{011}$ and $p_{10}=p_{100}+p_{101}$ to the distance between two uniform vectors is the same),
			\item the fact that $\sum_i \pvec_i = 1$.
		\end{itemize}
This gives us 7 equations for 8 variables leaving 1 degree of freedom. We further assume that $p_{000}=p_{111}$ (essentially, the same choice was done in \cite[Lemma2]{MO15}).
Solving these linear equations gives $\pvec_i$'s as stated in the theorem.  

It remains to argue that our choice for $|\C|$ also works for distances $\dist(\vvec, \qvec) \leq \gamma$. Informally, as $\gamma'<\gamma<1/2$, the Near Neighbour problem becomes easier as $\gamma'$ decreases.

Note that our choice for $\pvec$ constraints the choices for $\alpha, \beta, \gamma$ : these parameters should be chosen such that $\pvec_i \geq 0$, in particular, $\alpha-\beta<\gamma<\alpha+\beta$ (wlog.\ we assume $\alpha>\beta$). Let us fix $\alpha, \beta$. We show that for $0 \leq \gamma' \leq \gamma$, the probability (over the choice of $\cvec$) to find a solution is monotonously increasing as $\gamma$ decreases, and hence, $|\C|$-many vectors will suffice.  

Consider Eq.~\eqref{eq:InverseSizeOfC_proof} with $\gamma:=\gamma'$. Taking the derivative with respect to $\gamma'$, reveals that the probability (given by $1/|\C|$) increases as long as
\[
	 (1-2\beta)(2\alpha-1) \gamma'^2 - 2(\alpha-\beta)^2 \gamma' + (\alpha-\beta)^2 \leq 0.
\]

This is a quadratic equation in $\gamma'$ with roots at $\gamma'=(\alpha-\beta)/(1-2\beta)$ and $\gamma'=(\alpha-\beta)/(2\alpha-1)$ (the latter is negative as $\alpha>\beta$ and $\alpha \leq 1/2$). 
If $\gamma'$ satisfies $(\alpha-\beta)/(1-2\beta)< \gamma' <a+b$, the derivative is negative and we can choose $\alpha, \beta$ the same as we chose for $\gamma$, hence, with $\C$ defined by Eq.~\eqref{eq:InverseSizeOfC_proof}, we find all solutions. For $\gamma'<(\alpha-\beta)/(1-2\beta)$, we set $\alpha:=(1-2 \beta)\gamma'+\beta$ with $\beta$ as in Eq.~\eqref{eq:InverseSizeOfC_proof}. These parameters lead to $|\C| = 2^{n(1-H(\beta))}$, which depends only on (fixed) $\beta$ and, hence, cannot decrease with $\gamma$. $\hfill \qed$
\end{proof}

The next theorem gives asymptotical complexities of the routines in Alg.~\ref{alg:LSF}.

\begin{theorem}[LSF complexity for Hamming metric] \label{thm:LSFRunTime} 
	For the Near Neighbour problem with some fixed target $0 \leq \gamma \leq 1/2$ , the routines given in Alg.~\ref{alg:LSF} for some fixed $0 \leq \alpha, \beta \leq 1/2$ and the data structure $\D = \cup_\cvec \Bucket_{\cvec}$, have the following expected costs (up to terms sub-exponential in $n$):
	
	$\bullet$ Each $\Call{Update}{}$ costs \quad
	$
		T_{\mathtt{Upd}}^{\scriptscriptstyle \mathtt{LSF}} = |\C| \cdot  2^{(H(\alpha)-1) n}.
	$
	
	$\bullet$ Preprocessing 
	costs \quad
	$
		T_{\mathtt{Prep}}^{\scriptscriptstyle \mathtt{LSF}} = |\L| \cdot |\C| \cdot  2^{(H(\alpha)-1) n}.
	$
	
	$\bullet$ Each $\Call{Query}{}$ costs
	$
		T_{\mathtt{Query}}^{\scriptscriptstyle \mathtt{LSF}} = |\C| \cdot  2^{(H(\beta)-1) n} \cdot \EXPECT{|\Bucket_{\cvec}|},
	$ 
	where $	\EXPECT{|\Bucket_{\cvec}|}$ -- the expected size of each bucket -- is equal to $ |\L|\cdot 2^{(H(\alpha)-1) n}.$
\end{theorem}

\begin{proof}	
		We assume that our buckets $\C$ are constructed using `structured' filter vectors $\cvec$, which enables us  to find all the buckets within a certain distance to a fixed point in the output time (see the discussion above) and, at the same time, allows us to treat $\cvec$ as (sufficiently) uniform random vectors.
		
		The expected number of buckets for an update parameter $\alpha$ and a fixed  $\vvec$ is $|\C| \cdot \Pr_{\cvec \in \{0,1\}^n}{[\dist(\vvec, \cvec) = \alpha] } = |\C| \cdot \Pr{[\wt(\cvec-\vvec)=\alpha]} = |\C| \cdot 2^{H(\pvec(\alpha)) n} = |\C| \cdot 2^{(H(\alpha)-1) n}$.
		Preprocessing calls $\Call{Update}{\vvec}$ for all $\vvec \in \L$, hence its complexity is $|\L| \cdot T_{\mathtt{Upd}}$. 
		
		The probability that $\vvec$ will be added to a certain bucket during the update is again $2^{(H(\alpha)-1) n}$, so after $|\L|$ calls to $\Call{Update}{}$, the average bucket-load will be (up to sub-exponential terms) $|\L|\cdot 2^{(H(\alpha)-1) n}$. Treating $|\Bucket_{\cvec}|$ as a random variable and using standard Chernoff bound arguments, one can easily show that $|\Bucket_{\cvec}|$ does not significantly deviate from its expected value.
		
		During the $\Call{Query}{\qvec}$ calls, we find all $\beta$-close buckets in time $|\C| \cdot  2^{(H(\beta)-1) n}$ and for each bucket we look through $|\Bucket_{\cvec}|$ vectors and among them choose all $\gamma$-close to $\qvec$. $\hfill \qed$
	
\end{proof}

In the application to ISD, where the number of queries is equal to $|\L|$, it makes sense to setup the NN-parameters $\alpha$ and $\beta$ s.t.\ the time spent on preprocessing and the time spent on $|\L|$ queries are equal. Indeed, in May-Ozerov algorithm, we have $\alpha=\beta$ and, furthermore, $\alpha=H^{-1}(1-\log|\L|)$ to make the expected size of buckets equal to 1. 
After almost trivial algebraic manipulations with Eq.~\eqref{eq:NumOfBuckets} for these parameters, we obtain $\log |C| = (1-\gamma)\left(1- H\left( \frac{H^{-1}(1- \log|\L|) - \gamma/2 }{1-\gamma} \right)\right)$, which matches the result of \cite[Theorem1]{MO15}. 

\subsection{Quantum ISD with Near Neighbour}

Here we explain how to embed the Near Neighbour routines into quantum walk search.
In classical setting, we would create two lists  $\L_1, \L_2$ of equal size (see Fig.~\ref{fig:SternAndMMT}), setup the data structure $\D$ (i.e., choose enough filter-vectors) and call $\Call{Update}{\vvec_1}$ for all $\vvec_1 \in \L_1$ with some update parameter $\alpha$. This is the preprocessing stage. Then, for each $\vvec_2 \in \L_2$, we call $\Call{Query}{\vvec_2}$ for a query parameter $\beta$ and search through the output of size $|\C| \cdot  2^{(H(\beta)-1) n} \cdot \EXPECT{|\Bucket_{\cvec}|}$ for $\vvec_1 \in \D$ s.t.\ $\dist(\vvec_1, \vvec_2) = w-p$. This is the query stage. From the solution pair $(\vvec_1, \vvec_2)$, we retrieve the error-vector and solve the 2-list matching problem. If we set $\alpha=\beta$ to balance out the costs for updates and queries, and $\alpha = H^{-1}( 1 - \log |\L_1|)$ to balance preprocessing and query stages, we solve the $2$-list matching problem for ISD in time $|\C|$ which is exactly what May-Ozerov algorithm achieves.

It is not hard to see that the `Update-and-Query' description of the Near Neighbour search suits particularly well the quantum walk search framework. Assume we run the walk over a superposition of $(S_1, S_2) \subset \L_1 \times \L_2$, where $|S_1 \cup S_2| = \Theta(r)$ for some $r$ which will be determined later. During the Setup phase we create the LSF data structure for $S_1$, and call $\Call{Update}{\vvec_1}$ for all $\vvec_1 \in S_1$. Now, contrary to the classical setting, we apply Grover's search over all $\vvec_2 \in S_2$ with the Grover checking function being $\Call{Query}{\vvec_2}$, which tells us whether $(S_1, S_2)$ is `marked', i.e., whether it contains $(\vvec_1, \vvec_2)$ s.t.\ $\dist(\vvec_1, \vvec_2) = w-p$. This allows us to spend more time on $\Call{Query}{}$ calls choosing $\beta \neq \alpha$. 

We do the same during the Update+Check phases: we update $S_1 \cup S_2$ with $\Theta(\sqrt{r})$ new vectors and for each of them we call the LSF $\Call{Update}{}$ routine. We also delete $\Theta(\sqrt{r})$ vectors by calling the LSF  $\Call{Remove}{}$ function. The Checking phase of the walk calls $\Call{Query}$ in the superposition over all $\Theta(\sqrt{r})$ new vectors  and decides in time $\Theta(r^{1/4} \cdot T_{\mathtt{Query}}^{\scriptscriptstyle \mathtt{LSF}})$ whether the updated $S_1 \cup S_2$ is marked.

So the advantage of quantum walk is two-fold: first, we work only with exponentially shorter sublists $S_1, S_2$ and, second, during the Checking phase we use Grover over many $\Call{Query}{}$ calls. Alg.\ref{alg:QuantumISDWithNN}  below summarizes the above description and should be used at Step~(\ref{alg_line:Walk}) of Alg.~\ref{alg:MetaQuantumISD}. 

Finally, one can combine quantum Near Neighbour search with the Shamir-Schroeppel trick (see Sect.~\ref{sec:Prelims}): instead of working with long lists $\L_1, \L_2$, consider their sublists $\L_1', \L_2' \subset \L_1, \L_2$ s.t.\ $\vvec_1 \in \L_1'$ and $\vvec_2 \in \L_2'$ are equal to a certain vector $\tvec \in \GF^{\ell'}$ on $\ell'$-coordinates. The probability that $\L_1', \L_2'$ contain the solution is $2^{-\ell'}$. Quantumly, the cost to construct $\L_1', \L_2'$ that contain the solution is $2^{-\ell'/2}$ (Grover's search). 
Now run NN-search on shorter lists $\L_1', \L_2'$ and on the dimension reduced by $\ell'$. This adds one more parameter $\ell'$ into the optimization problem. Such algorithm offers a slight improvement both in time and memory over plain Stern's algorithm as the next theorem shows.

\begin{algorithm}[h]
	\caption{A quantum walk with Near Neighbour}
	\label{alg:QuantumISDWithNN}
	\vspace{5pt}
	\begin{algorithmic}[1]
		
	\Statex \hspace*{-10pt} Quantum walk SETUP:
		\State Create the LSF data structure $\D$ on the auxiliary register $\ket{Aux}$
		\ForAll{$\vvec_1 \in S_1$}    		\Comment{$|S_1| = |S_2| = \Theta(r)$}
			\State Call $\Call{Update}{\vvec_1}$ \Comment{Update $\D$}
		\EndFor
		\State Using Grover search over all $\vvec_2 \in S_2$: $\Call{Query}{\vvec_2}$ to check if $(S_1,S_2)$ is marked
		
	\vspace*{8pt} 
	\Statex \hspace*{-10pt} Quantum walk UPDATE:
		 \State $S_{\scriptscriptstyle \mathtt{new}} \gets \emptyset$ 
		 \State Repeat {$\Theta(\sqrt{|S_1|})$} times:
		 \State \hspace*{10pt} Call $\Call{Update}{\vvec^\star}$ \Comment{Add a new $\vvec^\star \notin S_1 \cup S_2$ to $\D$}
		 \State \hspace*{10pt} $S_{\scriptscriptstyle \mathtt{new}} \gets S_{\scriptscriptstyle \mathtt{new}} \cup \{\vvec^\star\}$
		 \State \hspace*{10pt} Call $\Call{Remove}{\vvec}$ \Comment{Remove $\vvec \in S_1 \cup S_2$ from $\D$} 
		 \State \hspace*{10pt} $(S_1, S_2) \gets (S_1,S_2) \setminus \{\vvec\}$
		 \State Update the register $(S_1, S_2)$ with $S_{\scriptscriptstyle \mathtt{new}}$

	\vspace*{8pt} 
	\Statex \hspace*{-10pt} Quantum walk CHECK:
	\State Run Grover search over all $\vvec_2 \in S_{\scriptscriptstyle \mathtt{new}}$ using $\Call{Query}{\vvec_2}$ to check if $(S_1,S_2)$ is marked

	\end{algorithmic}
\end{algorithm}


\begin{theorem}[Quantum Dumer+Near Neighbour]
	Assume we run Alg.~\ref{alg:MetaQuantumISD}  for Dumer's decoding, where during quantum walk we use the $(w-p)$-Near Neighbour routines from Alg.~\ref{alg:QuantumISDWithNN}. Then the expected running time of Dumer's algorithm is $\softO(2^{0.059922\cdot n + \smallo(n)})$ with quantum memory complexity $\softO(2^{0.00897\cdot n + \smallo(n)})$. Using additionally the Shamir-Schroeppel trick, time and memory can be improved to $\softO(2^{0.059450 \cdot n+\smallo(n)})$ and $\softO(2^{0.00808\cdot n+\smallo(n)})$.
\end{theorem}
\begin{proof}
	The number of trials $\P=\P(p, \ell)$ until we find a good permutation $\pi$ for the Near Neighbour version of Dumer's decoding is given in Eq.~\eqref{eq:PermutationProbabiltiyFS}. Grover's search will find a good $\pi$ in time $\bigO(\sqrt{\P})$. The checking routine for this search is a quantum walk over the subsets $(S_1, S_2) \subset \L_1 \times \L_2$ with $|\L_1| = |\L_2| = \binom{(k+\ell)/2}{p/2}$, where during the walk we look for an approximate match in $S_1 \cup S_2$ using Alg.~\ref{alg:QuantumISDWithNN}. Assume $|S_1 \cup S_2| = r$. We want to determine $r$ and the LSF parameters $\alpha, \beta$ for $\Call{Update}{}$ and $\Call{Query}{}$ that minimize the Near Neighbour search. In the following we omit the $\softO$-notation for all runtimes.
	
	The complexity of the quantum walk Setup  is $\max\{|\C|, r \cdot T_{\mathtt{Upd}}^{\scriptscriptstyle \mathtt{LSF}}, \sqrt{r} \cdot T_{\mathtt{Query}}^{\scriptscriptstyle \mathtt{LSF}} \}$, where $|\C|$ is given in Eq.~\eqref{eq:NumOfBuckets} and $T_{*}^{\scriptscriptstyle \mathtt{LSF}}$ is given in Thm.~\ref{thm:LSFRunTime}. That is, we take the maximum between the time to setup $\D$, call the $\Call{Update}{}$ $r$ times and run Grover over the $r$ new elements to decide on marked subsets for the starting superposition. The decision is realized via calling $\Call{Query}{}$.
	
	In the Update phase, we call $\sqrt{r}$ times $\Call{Update}{}$ and $\Call{Remove}{}$ LSF routines to update $\D$.
 	The complexity of the Update phase is  $\sqrt{r} \cdot T_{\mathtt{Upd}}^{\scriptscriptstyle \mathtt{LSF}}$, and of the checking phase is $r^{1/4} \cdot T_{\mathtt{Query}}^{\scriptscriptstyle \mathtt{LSF}}$.  As in the classical case, we set 
	\[ \alpha = H^{-1} (1- \log r)
	\] 
	to guarantee that the expected size of each bucket is 1.\footnote{One could also run Grover inside each bucket during the Query phase, when the buckets are larger than 1. This, however, does not seem to bring an improvement.}  Note that this choice also balances $|\C| = r \cdot T_{\mathtt{Upd}}^{\scriptscriptstyle \mathtt{LSF}}$ for the quantum walk Setup. 
	Finally, the quantum walk Checking routine runs Grover search over $\sqrt{r}$ new elements in $S_1 \cup S_2$ to update the `marking' flag for $S_1 \cup S_2$.
	To balance the Update and the Check phases (i.e, when $\sqrt{r} \cdot T_{\mathtt{Upd}} = r^{1/4} \cdot T_{\mathtt{Query}}$), we set 
	\[
	\beta = H^{-1}(1-\tfrac{3}{4} \log r). 
	\]
	Such choice also guarantees that during the Setup, $ r \cdot T_{\mathtt{Upd}}^{\scriptscriptstyle \mathtt{LSF}} \geq  \sqrt{r} \cdot T_{\mathtt{Query}}^{\scriptscriptstyle \mathtt{LSF}}$. Moreover, it enables us to setup  $\beta$ slightly larger than $\alpha$ since $\Call{Query}{}$  becomes cheaper. 
	
	Finally, we want to balance $T_{\mathtt{S}}$ for Setup, which is $r \cdot T_{\mathtt{Upd}}^{\scriptscriptstyle \mathtt{LSF}}$, with the Update and Check phases,  $\smash{\frac{1}{\sqrt{\eps}}\big(\frac{1}{\sqrt{\delta}} \cdot T_{\mathtt{U}} + T_{\mathtt{C}}\big) }$, cf.\ Eq.~\eqref{eq:QWRuntime}.  Due to our choices of $\alpha, \beta$, this expression is equal to $\frac{1}{\sqrt{\eps}} \cdot  \sqrt{r} T_{\mathtt{Upd}}^{\scriptscriptstyle \mathtt{LSF}}$ since $\delta \approx 1/r$. 
	
	For $k=2$, $\eps = r^2 /|\L_1 \cup \L_2|^2 $, from where we obtain
	\[
		r = |\L_1 \cup \L_2|^{2/3} \approx \binom{(k+\ell)/2}{p/2} ^{2/3}.
	\]
	
	The last parameter we need to determine in order to give the complexity of  decoding is the weight parameter $p$ for which we execute the $(w-p)$-Near Neighbour search. The brute-force search over $p$ reveals that for $p = 0.0027 \cdot n$, $\alpha = 0.4169 \cdot n$,   $\beta = 0.4280 \cdot n $, we have $|\C|=2^{0.00897 \cdot n}$. We obtain the figures stated in the theorem by computing the necessary number of permutation for such $p$ and noting that $|\C|$ determines the memory cost.
	
	If we construct the lists $\L_1, \L_2$ using the Shamir-Schroeppel idea, we start with $k=4$ list $\L_{i,j}$ each of size $\binom{(k+\ell)/4}{p/4}$. We merge them into 2 lists $\L_1, \L_2$ by enforcing the vectors $(\vvec_1, \vvec_2) \in \L_1 \times \L_2$  having the same value on $\ell'$ coordinates. Quantumly, we find the correct value for the ISD solution in time $2^{\ell'/2}$.
	We solve the 4-list matching problem via quantum walk with the optimal choice for $r = |\L_{i,j}|^{4/5}$. Optimization reveals that the choosing $p=0.043, \ell'=0.007, \alpha=0.4330, \beta=0.4419$, gives the best running time. $\hfill \qed$
\end{proof}
\paragraph{Why choosing larger $k$ does not help.}
	The more starting lists $k$ an ISD algorithm has, the larger the fraction $1/\eps = |L|^k/ r^k$ is  for any $r < |L|$. Hence, the running time of approximate search and, consequently, the running time of quantum walk become more expensive. The search for optimal parameters tries to shift the work-load to the Grover search for a good permutation by making $p$ smaller (the smaller $p$ is, the harder it is find a good $\pi$ but the easier the NN-search). From the above theorem, we have for $k=2$, $p = 0.0027$ which is already quite small. An optimization for  $k=4$ (e.g., MMT) chooses $p=0$ which is Prange's algorithm. This is also the reason why we do not get a quantum speed-up for algorithms proposed in \cite{BM17}.

\paragraph{Acknowledgements.} The author thanks Alexander May for enlightening discussions and suggestions and A.Helm for detecting a mistake in the first version of the paper. This work is supported by ERC Starting Grant ERC-2013-StG-335086-LATTAC.


\def\shortbib{1}
\bibliographystyle{alpha}
\bibliography{mybib}

\end{document}